\documentclass{amsart}

\usepackage{amssymb}
\usepackage{amsthm}
\usepackage{wasysym}
\usepackage[numbers]{natbib}

\newcommand{\abs}[1]{\left|#1\right|}
\newcommand{\norm}[1]{\left\| #1 \right\|}
\newcommand{\normLzwei}[1]{\left\| #1 \right\|_{L^2}}
\newcommand{\OpNormLzwei}[1]{\left\| #1 \right\|_{L^2\rightarrow L^2}}

\newcommand{\Ho}{\mathcal{H}}

\newtheorem{Theorem}{Theorem}
\newtheorem{Proposition}[Theorem]{Proposition}
\newtheorem{Lemma}[Theorem]{Lemma}
\newtheorem{Remark}[Theorem]{Remark}

\newcommand{\R}{\mathbb{R}}
\newcommand{\C}{\mathbb{C}}
\newcommand{\N}{\mathbb{N}}

\begin{document}

\title[Optimal decay rate for the damped wave equation]{Optimal decay rate for the wave equation on a square with constant damping on a strip}
\author[Reinhard Stahn]{Reinhard Stahn}

\begin{abstract} 
We consider the damped wave equation with Dirichlet boundary conditions on the unit square parametrized by Cartesian coordinates $x$ and $y$. We assume the damping $a$ to be strictly positive and constant for $x<\sigma$ and zero for $x>\sigma$. We prove the exact $t^{-4/3}$-decay rate for the energy of classical solutions. Our main result (Theorem \ref{thm: goal}) answers question (1) of \citep[Section 2C.]{AnantharamanLeautaud2014}.
\end{abstract}

\maketitle

{\let\thefootnote\relax}
{\let\thefootnote\relax}


\section{Introduction}\label{sec: Introduction}

\subsection{The main result}\label{sec: The main result}
Let $\Box = (0,1)^2$ be the unit square. We parametrize it by Cartesian coordinates $x$ and $y$. Let $a$ - the damping - be a function on $\Box$ which depends only on $x$ such that $a(x)=a_0>0$ for $x<\sigma$ and $a(x)=0$ for $x>\sigma$ where $\sigma$ is some fixed number from the interval $(0,1)$. We consider the damped wave equation:
  \begin{equation}\nonumber
    \begin{cases}
      u_{tt}(t,x,y)-\Delta u(t,x,y) + 2a(x) u_t(t,x,y) = 0 & (t\in(0,\infty),\, (x,y)\in \Box), \\
      u(t,x,y) = 0 & (t\in(0,\infty),\, (x,y)\in\partial\Box), \\
      u(0,x,y) = u_0(x,y),\, u_t(0,x,y) = u_1(x,y) & ((x,y)\in \Box) .
    \end{cases}
  \end{equation}
We are interested in the energy
  \begin{equation}\nonumber
    E(t,U_0) = \frac{1}{2} \int\int \abs{\nabla u(t,x,y)}^2 + \abs{u_t(t,x,y)}^2 \, dx dy   
  \end{equation}
of a wave at time $t$ with initial data $U_0=(u_0, u_1)$. Let $D = (H^2\cap H^1_0) \times H^1_0(\Box)$ denote the set of classical initial data. The purpose of this paper is to prove
  \begin{Theorem}\label{thm: goal}
    Let $\Box$, $a$ and $E(t,U_0)$ be as above. Then $\sup E(t,U_0)^{1/2} \approx t^{-2/3}$ where the supremum is taken over initial data $\norm{U_0}_{D}=1$.
  \end{Theorem}
The exact meaning of `$\approx$` and other symbols is explained in Section \ref{sec: Notation}. In Section \ref{sec: Exact Decay Rate for the Damped Wave Equation} we show that this theorem is equivalent to Theorem \ref{thm: Maintheorem} below. Section \ref{sec: Proof of Theorem} is devoted to the proof of Theorem \ref{thm: Maintheorem}.
  \begin{Remark}
    The proof of Theorem \ref{thm: goal} shows that a higher dimensional analogue is also true. That is, one can replace $y\in\R$ by $y\in\R^{d-1}$ for any natural number $d\geq 2$. The exact decay rate remains the same for all $d$.
  \end{Remark}
  
\subsection{The semigroup approach}
If we set $U=(u,u_t)$ and $U_0=(u_0,u_1)$ we may formulate the damped wave equation as an abstract Cauchy problem
  \begin{align*}
    \dot{U}(t)+AU(t)=0,\, U(0)=U_0\, \text{ where }\,
    A = 
    \begin{pmatrix}
      0       & -1 \\
      -\Delta & 2a(x)
    \end{pmatrix}
  \end{align*}
on the Hilbert space $\Ho=H^1_0\times L^2(\Box)$. The domain of $A$ is $D(A) = (H^2\cap H^1_0) \times H^1_0(\Box)$. Since $-A$ is a dissipative (we equip $H^1_0(\Box)$ with the gradient norm) and invertible operator on a Hilbert space it generates a $C_0$-semigroup of contractions by the Lumer-Phillips theorem. Note that the inclusion $D(A)\hookrightarrow \Ho$ is compact by the Rellich-Kondrachov theorem. Thus the spectrum of $A$ contains only eigenvalues of finite multiplicity.

\subsection{Classification of the main result}
Our situation is a very particular instance of the so called \emph{partially rectangular} situation. A bounded domain $\Omega$ is called \emph{partially rectangular} if its boundary $\partial\Omega$ is piecewise $C^{\infty}$ and if $\Omega$ contains an open rectangle $R$ such that two opposite sides of $R$ are contained in $\partial\Omega$. We call these two opposite sides \emph{horizontal}. One can decompose $\overline{\Omega}=\overline{R}\cup \overline{W}$, where $W$ is an open set which is disjoint to $R$. In our particular situation we can $W$ choose to be empty. Furthermore it is assumed, that $a>0$ on $\overline{W}$ and $a=0$ on $S$, where $S\subseteq R$ is an open rectangle with two sides contained in the horizontal sides of $R$. To avoid the discussion of null-sets we assume for simplicity that either $a$ is continuous up to the boundary or it is as in subsection \ref{sec: The main result}.  

Under these constraints one can show that the energy of classical solutions can never decay uniformly faster than $1/t^2$, i.e.
  \begin{equation}\label{eq: fast bound}
    \sup_{U_0\in D(A)} E(t, U_0)^{\frac{1}{2}} \gtrsim \frac{1}{t}.
  \end{equation}
This result seems to be well-known. Unfortunately we do not know an original reference to this bound on the energy. A short modern proof using \citep[Proposition 1.3]{BattyDuyckaerts2008} can be found in \citep{AnantharamanLeautaud2014}. But there is also a \emph{geometric optics} proof using quantified versions of the techniques of \citep{Ralston1969}. Unfortunately the latter approach seems to be never published anywhere.

On the other hand: If we assume that the damping does not vanish completely in $R$ (this is an additional assumption only if $W$ is empty), then 
  \begin{equation}\label{eq: slow bound}
    \forall U_0\in D(A): E(t, U_0)^{\frac{1}{2}} \lesssim \frac{1}{t^{\frac{1}{2}}}.
  \end{equation}
This is a corollary of one of the main results in \citep{AnantharamanLeautaud2014}. There the authors showed that \emph{stability at rate} $t^{-1/2}$ for an \emph{abstract} damped wave equation is equivalent to an observability condition for a related Schr\"{o}dinger equation. Earlier contributions towards (\ref{eq: slow bound}) were given by \citep{BurqHitrik2007} and \citep{LiuRao2005}.

Having the two bounds (\ref{eq: fast bound}) and (\ref{eq: slow bound}) at hand a natural question arises: Are these bounds sharp? Concerning the fast decay rates related to (\ref{eq: fast bound}) this is partly answered by \citep{BurqHitrik2007} and \citep{AnantharamanLeautaud2014}. Essentially the authors showed that if the damping function is smooth enough than one can get a decay rate as close to $t^{-1}$ as we wish. Unfortunately they could not \emph{characterize} the \emph{exact} decay rate in terms of properties of $a$. A breakthrough into this direction was achieved in \citep{LeautaudLerner2014} in a slightly different situation (there $S$ degenerates to a line).

To the best of our knowledge it is completely unknown if the slowest possible rate $t^{-1/2}$ is attained. To us the only known result towards this direction is due to Nonnenmacher: If we are in the very particular situation described in subsection \ref{sec: The main result} then
  \begin{equation}\nonumber 
    \sup_{U_0\in D(A)} E(t, U_0)^{\frac{1}{2}} \gtrsim \frac{1}{t^{\frac{2}{3}}}.
  \end{equation}
See \citep[Appendix B]{AnantharamanLeautaud2014}. So this situation is a candidate for the slow decay rate. In this paper we show that Nonnenmacher's bound is actually equal to the exact decay rate. 

This of course raises a new question: Is it possible to find a non-vanishing bounded damping in a partially rectangular domain, satisfying the constraints specified above, but discarding the continuity assumptions, such that the exact decay rate for $E(t, U_0)^{\frac{1}{2}}$ is strictly slower than $t^{-2/3}$? We think this is an interesting question for future research.

\subsection{From waves to stationary waves}
Let $f\in L^2(\Box)$. Now we consider the stationary damped wave equation with Dirichlet boundary conditions
  \begin{equation}\label{eq: Stationary Wave Equation}
    \left\{
    \begin{array}{rl}
      P(s)u(x,y) = (-\Delta - s^2 + 2isa(x))u(x,y) = f(x,y) & \text{in } \Box \\
      u(x,y) = 0 & \text{on } \partial\Box
    \end{array}
    \right.
  \end{equation}
As already said above, to prove Theorem \ref{thm: goal} is essentially to show
  \begin{Theorem}\label{thm: Maintheorem}
    The operator $P(s):H^2\cap H^1_0(\Box)\rightarrow L^2(\Box)$ from (\ref{eq: Stationary Wave Equation}) is invertible for every $s\in\R$. Moreover
      \begin{equation}\nonumber
        \OpNormLzwei{P(s)^{-1}} \approx 1 + \abs{s}^{\frac{1}{2}} .
      \end{equation}
  \end{Theorem}
Actually we only prove a $\lesssim$-inequality since the reverse inequality is a consequence of Nonnenmacher's appendix to \citep{AnantharamanLeautaud2014} together with Proposition 2.4 in the same paper (see Section \ref{sec: Exact Decay Rate for the Damped Wave Equation} for more details). Since it is well-known we also do not prove the invertability of $P(s)$. The (simple) standard proof is based on testing the homogeneous stationary wave equation with $\overline{u}$. From considering real and imaginary part of the resulting expression one easily checks $u=0$ by a \emph{unique continuation principle}. 

\subsection*{Acknowledgments} This paper was inspired and motivated by \citep[Appendix B (by S. Nonnenmacher)]{AnantharamanLeautaud2014} and \citep{BattyPaunonenSeifert2015}. I am grateful to Ralph Chill for reading and correcting the very first version of this paper.

\section{Notations and conventions}\label{sec: Notation}
\emph{Convention}. Because of the symmetry of (\ref{eq: Stationary Wave Equation}) we have $\OpNormLzwei{P(-s)^{-1}}=\OpNormLzwei{P(s)^{-1}}$. Therefore in the following we always assume $s$ to be \emph{positive}.

\emph{Constants}. We use two special constants $c>0$ and $C>0$. Special means, that they may change their value from line to line. The difference between these two constants is, that their usage implicitly means that we could always replace $c$ by a smaller constant and $C$ by a larger constant - \emph{if this is necessary}. So one should keep in mind that $c$ is a small number and $C$ a large number. 

\emph{Landau notation}. For this subsection let us denote by $\phi,\phi_1,\phi_2$ and $\psi$ complex valued functions defined on $\R\backslash K$, where $K$ is a compact interval. Furthermore we always assume $\phi,\phi_1$ and $\phi_2$ to be real valued and (not necessary strictly) positive.   
We define
  \begin{align*}
    \phi_1(s) \lesssim \phi_2(s) &:\Leftrightarrow \exists s_0>0, C>0\forall\abs{s}\geq s_0: \phi_1(s)\leq C\phi_2(s), \\
    \phi_1(s) \approx \phi_2(s)  &:\Leftrightarrow \phi_1(s) \lesssim \phi_2(s) \text{ and } \phi_2(s) \lesssim \phi_1(s).
  \end{align*}
Furthermore we define the following classes (sets) of functions:
  \begin{align*}
    O(\phi(s)) &:= \{\psi; \abs{\psi(s)}\lesssim\phi(s)\}, \\
    o(\phi(s)) &:= \{\psi; \forall\varepsilon>0 \exists s_{\varepsilon}>0\forall\abs{s}\geq s_{\varepsilon}: \abs{\psi(s)}\leq \varepsilon\phi(s)\}.
  \end{align*}
By abuse of notation we write for example $\psi(s)=O(\phi(s))$ instead of $\psi\in O(\phi(s))$ or $\phi(s)= \phi_1(s) + O(\phi_2(s))$ instead of $\abs{\phi(s)-\phi_1(s)} \lesssim \phi_2(s)$. By $O(s^{-\infty})$ we denote the intersection of all $O(s^{-N})$ for $N\in \N$.

\emph{Function spaces.} As usual, by $L^2(\Omega)$ we mean the space of square-integrable functions on some open subset $\Omega$ of $\R^n$ for some $n\in\N$. For $k$ a natural number $H^k(\Omega)$ denotes the space of functions from $L^2(\Omega)$ whose distributional derivatives up to order $k$ are square integrable, too. Finally the space $H^1_0(\Omega)$ denotes the closure of the set of compactly supported smooth functions in $H^1(\Omega)$. We equip $H^1_0(\Omega)$ with the norm $(\int_{\Omega}\abs{\nabla u}^2dx)^{1/2}$ which is equivalent to the usual norm.

\section{Proof of Theorem \ref{thm: Maintheorem}}\label{sec: Proof of Theorem}
Here is the plan for the proof: First we separate the $y$-dependence of the stationary wave equation from the problem. As a result we are dealing with a family of one dimensional problems which are parametrized by the vertical wave number $n\in\N$. Then we derive explicit solution formulas for the separated problems. These formulas allow us to estimate the solutions of the separated problems by their right-hand side with a constant essentially depending \emph{explicitly} on $s$ and $n$. In the final step we introduce appropriate regimes for $s$ relative to $n$ which allow us to drop the $n$-dependence of the constant by a (short) case study.

\subsection{Separation of variables}\label{sec: Separation of Variables}
First recall that the functions $s_n(y)=\sqrt{2}\sin(n\pi y)$ for $n\in\{1,2,\ldots\}$ form a complete orthonormal system of $L^2(0,1)$. Thus considering $u$ and $f$ satisfying (\ref{eq: Stationary Wave Equation}) we may write
  \begin{align}\label{eq: Separation of u and f}
    u(x,y) = \sum_{n=1}^{\infty} u_n(x)s_n(y) \text{ and }
    f(x,y) = \sum_{n=1}^{\infty} f_n(x)s_n(y).
  \end{align}
In terms of this separation of variables the stationary wave equation is equivalent to the one dimensional problem $P_n(s)u_n=f_n$ where
  \begin{align}
    P_n(s) = -\partial_x^2 - k_n^2 + 2isa(x), \text{ and} \label{eq: separated Stationary Wave Operator} \\ \nonumber
    k_n^2 = s^2 - (n\pi)^2 .
  \end{align}
Note that $k_n$ might be an imaginary number. In a few lines we see that only the real case is important. In that case we choose $k_n\geq0$. But first we prove the following simple
  \begin{Lemma}\label{thm: P to Pn Lemma}
    Let $\phi:\R\rightarrow (0,\infty)$. Then the estimate $\OpNormLzwei{P_n(s)^{-1}}\lesssim\phi(s)$ uniformly in $n$ is equivalent to the estimate $\OpNormLzwei{P(s)^{-1}}\lesssim\phi(s)$.
  \end{Lemma}
  \begin{proof}
    Let $P(s)u=f$ and expand $u$ and $f$ as in (\ref{eq: Separation of u and f}). Then the implication from the left to the right is a consequence of the following chain of equations and inequalities:
    \begin{align*}
      \normLzwei{u}^2 &= \sum_{n=1}^{\infty} \normLzwei{u_n}^2 
      \lesssim \phi(s)^2 \sum_{n=1}^{\infty} \normLzwei{f_n}^2 
      = \phi(s)^2 \normLzwei{f}^2.
    \end{align*}
    The reverse implication follows from looking at $f(x,y)=f_n(x)s_n(y)$ and $u(x,y)=u_n(x)s_n(y)$. 
  \end{proof}
So below we are concerned with the separated stationary wave equation
  \begin{equation}\label{eq: separated Stationary Wave Equation}
    \left\{
    \begin{array}{rl}
      P_n(s)u_n(x) = f_n(x) & \text{for } x\in(0,1) \\
      u_n(0) = u_n(1) = 0 &
    \end{array}
    \right.
  \end{equation}
where $P_n(s)$ is defined in (\ref{eq: separated Stationary Wave Operator}). In view of Lemma \ref{thm: P to Pn Lemma} we are left to show $\normLzwei{u_n}\lesssim s^{1/2}\normLzwei{f_n}$ uniformly in $n$ in order to prove Theorem \ref{thm: Maintheorem}. It turns out that such an estimate is easy to prove if $k_n$ is imaginary. More precisely:
  \begin{Lemma}\label{thm: forbidden s}
    There exists a constant $c>0$ such that $\norm{P_n(s)^{-1}}_{L^2\rightarrow H^1_0} \lesssim 1$ holds uniformly in $n$ whenever $s^2 \leq (n\pi)^2 + c$.
  \end{Lemma}
  Note that $P_n(s)^{-1}$ is considered as an operator mapping to $H^1_0(0,1)$. But it does not really matter since we will only use this estimate after replacing $H^1_0$ by $L^2$.
  \begin{proof}
  Testing equation (\ref{eq: separated Stationary Wave Equation}) by $\overline{u}_n$ and taking the real part leads to
    \begin{equation}\nonumber
      \int_0^1\abs{u'_n}^2 - c\int_0^1\abs{u_n}^2 \leq \int_0^1 \abs{f_n u_n} .
    \end{equation}
  Recall that $\normLzwei{v'}^2\geq\pi^2\normLzwei{v}^2$ for all $v\in H^1_0(0,1)$ since $\pi^2$ is the lowest eigenvalue of the Dirichlet-Laplacian on the unit interval. Thus the conclusion of the Lemma holds for all $c<\pi^2$.
  \end{proof}
This lemma allows us to assume
  \begin{equation}\label{eq: k is real}
    k_n = \sqrt{s^2 - (n\pi)^2} > c
  \end{equation}
for some universal constant $c>0$ not depending on neither $s$ nor $n$. 

\subsection{Explicit formula for $P_n(s)^{-1}$}\label{sec: Explicit formula for P_n(s)^{-1}} 
From now on we consider (\ref{eq: separated Stationary Wave Equation}) under the constraint (\ref{eq: k is real}). To avoid cumbersome notation we drop the subscript $n$ from $k_n$, i.e. we write $k$ instead from now on. Next let $v=u_n|_{[0,\sigma]}, g=f_n|_{(0,\sigma)}$ and $w=u_n|_{[\sigma,1]}, h=f_n|_{(\sigma,1)}$. We may write (\ref{eq: separated Stationary Wave Equation}) as a coupled system consisting of a wave equation with constant damping and an undamped wave equation:
  \begin{equation}\label{eq: splitted and separated Stationary Wave Equation}
    \left\{
    \begin{array}{rl}
      (-\partial_x^2-k^2+2isa_0)v(x)  = g(x) & \text{for } x\in(0,\sigma),    \\
      (-\partial_x^2-k^2)w(x) = h(x)         & \text{for } x\in(\sigma, 1),   \\
      v(0) = w(1) = 0,                       &                                \\
      v(\sigma) = w(\sigma), v'(\sigma)=w'(\sigma). &
    \end{array}
    \right.
  \end{equation}

\subsubsection{Solution of the homogeneous equation} The following ansatz satisfies the first three lines of (\ref{eq: splitted and separated Stationary Wave Equation}) with $g,h=0$:
  \begin{align}\label{eq: Def of v_0 and w_0}
    v_0(x) = \frac{1}{k'}\sin(k'x), \quad
    w_0(x) = \frac{1}{k}\sin(k(1-x)),
  \end{align}
where $k'$ is the solution of $k'^2 = k^2-2isa_0$ which has negative imaginary part.

\subsubsection{Solution of the inhomogeneous equation} The following ansatz satisfies the first three lines of (\ref{eq: splitted and separated Stationary Wave Equation}):
  \begin{align}\label{eq: Def of v_g and w_h}
    v_g(x) = -\frac{1}{k'}\int_0^x \sin(k'(x-y))g(y) dy, \,
    w_h(x) = -\frac{1}{k}\int_x^1 \sin(k(y-x))h(y) dy.
  \end{align}
This is simply the variation of constants (or Duhamel's) formula. It is useful to know the derivatives of these particular solutions:
  \begin{align}\label{eq: Derivative of v_g and w_h}
    v_g'(x) = -\int_0^x \cos(k'(x-y))g(y) dy, \,
    w_h'(x) = +\int_x^1 \cos(k(y-x))h(y) dy.
  \end{align}

\subsubsection{General solution} The general solution of the first three lines of (\ref{eq: separated Stationary Wave Equation}) has the form
  \begin{align}\label{eq: Def of v and w}
    v = av_0 + v_g, \quad w = bw_0 + w_h .
  \end{align}
Our task is to find the coefficients $a=a(s,n)$ and $b=b(s,n)$. Therefore we have to analyze the coupling condition in line four of (\ref{eq: splitted and separated Stationary Wave Equation}). A short calculation shows that it is equivalent to
  \begin{align*}
    \underbrace{
      \left.
      \begin{pmatrix}
        v_0  & -w_0  \\
        v_0' & -w_0'
      \end{pmatrix}
      \right|_{x=\sigma}
    }_{=:\, M(s,n)}
    \begin{pmatrix}
      a \\
      b
    \end{pmatrix}
    =
    \left.
    \begin{pmatrix}
      w_h  - v_g  \\
      w_h' - v_g'
    \end{pmatrix}
    \right|_{x=\sigma}
    .
  \end{align*}
From the preceding equation we easily deduce
  \begin{align}
    \label{eq: a}
    a &= \frac{1}{\det M}\left[w_0'(v_g-w_h) - w_0(v_g'-w_h)\right]_{x=\sigma},   \\
    \label{eq: b}
    b &= \frac{1}{\det M}\left[v_0'(v_g-w_h) - v_0(v_g'-w_h)\right]_{x=\sigma}.
  \end{align}
Moreover
  \begin{equation}\label{eq: detM}
    \det M = \frac{1}{k'}\sin(k'\sigma)\cos(k(1-\sigma) + \frac{1}{k}\cos(k'\sigma)\sin(k(1-\sigma))) .
  \end{equation}
  
\subsection{Proving a general estimate $\normLzwei{u_n}\leq C(k,k',M)\normLzwei{f_n}$}\label{sec: Proving a general estimate}
For this inequality we will derive an \emph{explicit} formula for $C$ in terms of $k, k'$ and $M$. In the next subsection we identify the qualitatively different regimes in which $s$ can live. By \emph{regime} we mean a relation which says how big $s$ - the full momentum - is compared to $n\pi$ - the momentum in $y$-direction. For each of these regimes we then easily translate the \emph{explicit} $k,k',M$ dependence of $C$ to a an \emph{explicit} dependence on $s$.

\subsubsection{Elementary estimates for $w_0$ and $ w_h$}
Directly from the definition of $w_0$ (see (\ref{eq: Def of v_0 and w_0})) we deduce
  \begin{equation}\label{eq: Elementary estimate for w_0}
    \norm{w_0}_{\infty} \leq \frac{1}{k}, \,
    \norm{w_0'}_{\infty} \leq 1 \text{ and }
    \norm{w_0}_{2} \leq \frac{\sqrt{1-\sigma}}{k}. 
  \end{equation}
In the same manner for $w_h$ from (\ref{eq: Def of v_g and w_h}) and (\ref{eq: Derivative of v_g and w_h}) we deduce:
  \begin{equation}\label{eq: Elementary estimate for w_h}
    \norm{w_h}_{\infty} \leq \frac{\sqrt{1-\sigma}}{k}\norm{h}_2, \,
    \norm{w_h'}_{\infty} \leq \sqrt{1-\sigma}\norm{h}_2 \text{ and }
    \norm{w_h}_{2} \leq \frac{1-\sigma}{k}\norm{h}_2. 
  \end{equation}

\subsubsection{Estimating $w$}
Recall from (\ref{eq: Def of v and w}) that $w=bw_0+w_h$. Recall the formula (\ref{eq: b}) for $b$. Note that 
  \begin{equation}\nonumber
    (v_0'v_g - v_0v_g')(\sigma) = \frac{1}{k'}\int_0^{\sigma} \sin(k'y) g(y) dy .
  \end{equation}
Thus it seems to be natural to decompose
  \begin{align*}
    b &= \frac{1}{\det M}\left[(v_0w_h'-v_0'w_h) + (v_0'v_g-v_0v_g')\right]_{x=\sigma} \\
         &=: b_1 + b_2 .
  \end{align*}
This leads to the decomposition of $w=b_1w_0 + b_2w_0 + w_h$ into three parts. With the help of (\ref{eq: Elementary estimate for w_0}) and (\ref{eq: Elementary estimate for w_h}) each part can easily be estimated as follows:
  \begin{equation}
    \begin{aligned}\label{eq: Estimate on w}
      \norm{b_1w_0}_2 \lesssim \frac{e^{\abs{\Im k'}\sigma}}{\abs{k' \det M}} \left(\frac{1}{k} + \frac{\abs{k'}}{k^2}\right) \norm{h}_2 ,   \\ 
      \norm{b_2w_0}_2 \lesssim \frac{e^{\abs{\Im k'}\sigma}}{\abs{k' \det M}} \frac{1}{k} \norm{g}_2 , \,                                     
      \norm{w_h}_2    \lesssim \frac{1}{k} \norm{h}_2 .
    \end{aligned}
  \end{equation}
We could now add all three single estimates to get the desired estimate on $w$ but we wait until we have done the same thing for $v$. 

\subsubsection{Estimating v}
Recall from (\ref{eq: Def of v and w}) that $v=av_0+v_h$. Recall the formula (\ref{eq: a}) for $a$. Note that 
  \begin{align*}
    (w_0w_h' - w_0'w_h)(\sigma) = \frac{1}{k}\int_{\sigma}^1 \sin(k(1-y)) h(y) dy \text{ and} \\
    v_g = \frac{(-w_0'v_0 + w_0v_0')(\sigma)}{\det M} v_g =: v_{g,2} + v_{g,3} .
  \end{align*}
Thus it seems to be natural to decompose
  \begin{align*}
    a &= \frac{1}{\det M}\left[(w_0w_h'-w_0'w_h) + w_0'v_g - w_0v_g'\right]_{x=\sigma} \\
         &=: a_1 + a_2 + a_3 .
  \end{align*}
This in turn leads to a decomposition of $v = a_1v_0 + (a_2v_0+v_{g,2}) + (a_3v_0+v_{g,3})$ into three parts. Essentially it leaves to find a good representation of the second and the third part of $v$. First let us write
  \begin{align*}
    a_2v_0 + v_{g,2} &= \frac{w_0'(\sigma)}{k'\det M} \underbrace{\left(v_g(\sigma)\sin(k'x) - k'v_0(\sigma)v_g(x)\right)}_{=:\, I(x)} , \\
    a_3v_0 + v_{g,3} &= \frac{w_0(\sigma)}{k'\det M} \underbrace{\left(-v_g'(\sigma)\sin(k'x) + k'v_0'(\sigma)v_g(x)\right)}_{=:\, II(x)} .
  \end{align*}
Simple calculations yield
  \begin{align*}
    -2I(x) =& \int_0^{\sigma} \cos(k'(\sigma-x-y))g(y) dy - \int_0^x \cos(k'(\sigma-x+y))g(y) dy \\
            &- \int_x^{\sigma} \cos(k'(\sigma+x-y))g(y) dy ,
  \end{align*}
and
  \begin{align*}
    2II(x) =& \int_x^{\sigma} \sin(k'(\sigma+x-y))g(y) dy - \int_0^x \sin(k'(\sigma-x+y))g(y) dy \\
            &- \int_0^{\sigma} \sin(k'(\sigma-x+y))g(y) dy .
  \end{align*}
Using this and again the elementary estimates (\ref{eq: Elementary estimate for w_0}) and (\ref{eq: Elementary estimate for w_h}) for $w_0$ and $w_h$ we deduce
  \begin{equation}
    \begin{aligned}\label{eq: Estimate on v}
      \norm{a_3v_0+v_{g,3}}_2 \lesssim \frac{e^{\abs{\Im k'}\sigma}}{\abs{k' \det M}} \frac{1}{k} \norm{g}_2 ,   \\ 
      \norm{a_2v_0+v_{g,2}}_2 \lesssim \frac{e^{\abs{\Im k'}\sigma}}{\abs{k' \det M}} \norm{g}_2 , \,                                     
      \norm{a_1v_0}_2 \lesssim \frac{e^{\abs{\Im k'}\sigma}}{\abs{k' \det M}} \frac{1}{k} \norm{h}_2 .
    \end{aligned}
  \end{equation}

\subsubsection{Conclusion}
Putting (\ref{eq: Estimate on w}) and (\ref{eq: Estimate on v}) together we get the desired inequality
  \begin{equation}\label{eq: Preliminary Resolvent Estimate}
    \normLzwei{u_n}\lesssim \left[ \frac{e^{\abs{\Im k'}\sigma}}{\abs{k'\det M}} \left(1+\frac{\abs{k'}}{k^2}\right) + \frac{1}{k} \right]\normLzwei{f_n}.
  \end{equation}

\subsection{Regimes where $s$ can live}\label{Regimes where s can live}
Keeping (\ref{eq: Preliminary Resolvent Estimate}) in mind, our task is now to find asymptotic dependencies of $k$ and $k'$ on $s$ and a lower bound for $\abs{k'\det M}$. A priori there is no unique asymptotic behavior of $k=\sqrt{s^2-(n\pi)^2}$ as $s$ tends to infinity because of $k$'s dependence on $n$. To overcome this difficulty we introduce the following four \emph{regimes}:
  \begin{equation}\nonumber
    \text{(i) }   c \leq k \leq c s^{\frac{1}{2}} ,\,
    \text{(ii) }  c s^{\frac{1}{2}} \leq k \leq C s^{\frac{1}{2}} ,\,
    \text{(iii) } C s^{\frac{1}{2}} \leq k \leq c s ,\,
    \text{(iv) }  c s \leq k < s  .
  \end{equation}
Recall from Section \ref{sec: Notation} that $c$ (resp. $C$) means a small (resp. big) number. Both constants may be different in each regime. But by the convention made in section \ref{sec: Notation} we may assume that consecutive regimes overlap.

Since we want to investigate the asymptotics $s\rightarrow\infty$ we always may assume $s>s_0$ for some sufficiently large number $s_0>0$.

\subsubsection{Regime (i): $c \leq k \leq c s^{\frac{1}{2}}$}
For sufficiently small $c$ the first order Taylor expansion of the square root at $1$ gives a good approximation of
  \begin{equation}\nonumber
    k' = \sqrt{2a_0} s^{\frac{1}{2}} e^{-\frac{i\pi}{4}} \left(1 + \frac{ik^2}{a_0 s} + O(k^4 s^{-2})\right) .
  \end{equation}
In particular $\Im k' = -\sqrt{a_0} s^{\frac{1}{2}} (1 + O(k^2 s^{-1}))$ tends with a polynomial rate to minus infinity as $s$ tends to infinity. Therefore $\cot(k'\sigma) = i + O(s^{-\infty})$. Together with (\ref{eq: detM}) this gives us the following useful formula for
  \begin{equation}\label{det M (i)}
    \det M = \frac{\sin(k'\sigma)}{k'}\left[ \cos(k(1-\sigma)) + \frac{k'}{k}(i+O(s^{-\infty}))\sin(k(1-\sigma)) \right] .
  \end{equation}
It is not difficult to see that the term within the brackets is bounded away from zero. Thus $\abs{k' \det M} \gtrsim \exp(\abs{\Im k'}\sigma)$. From (\ref{eq: Preliminary Resolvent Estimate}) now follows (recall also (\ref{eq: k is real}))
  \begin{equation}\nonumber
    \normLzwei{u_n}\lesssim \left(1+\frac{\abs{k'}}{k^2}\right) \normLzwei{f_n} \lesssim s^{\frac{1}{2}} \normLzwei{f_n} \text{ uniformly in } n.
  \end{equation}

\subsubsection{Regime (ii): $c s^{\frac{1}{2}} \leq k \leq C s^{\frac{1}{2}}$}
Because of $k'^2 = k^2 - 2isa_0$ we see that both $\Re k'$ and $-\Im k'$ are of order $s^{\frac{1}{2}}$. Therefore (\ref{det M (i)}) is valid also in this regime. Again the term within the brackets is bounded away from zero. Thus $\abs{k' \det M} \gtrsim \exp(\abs{\Im k'}\sigma)$ and (\ref{eq: Preliminary Resolvent Estimate}) imply
  \begin{equation}\nonumber
    \normLzwei{u_n} \lesssim \normLzwei{f_n} \text{ uniformly in } n.
  \end{equation}

\subsubsection{Regime (iii): $C s^{\frac{1}{2}} \leq k \leq c s$}
Using first order Taylor expansion for the square root at $1$ gives
  \begin{equation}\nonumber
    k' = k \left(1 - ia_0 sk^{-2} + O(s^2 k^{-4})\right) .
  \end{equation}
In particular: If we choose $C$ big enough we can assume the ratio $k'/k$ to be as close to $1$ as we wish. Similarly: If we choose $c$ small enough we may assume $-\Im k'$ to be as large as we want. Therefore we may assume $\cot(k'\sigma)$ to be as close to $i$ as we wish. This means that the following variant of (\ref{det M (i)}) is true for this regime
  \begin{equation}\nonumber
    \det M = \frac{\sin(k'\sigma)}{k'}\left[ \cos(k(1-\sigma)) + (i+\varepsilon)\sin(k(1-\sigma)) \right] ,
  \end{equation}
where $\varepsilon\in\C$ is some error term with a magnitude as small as we wish. If we choose $c$ and $C$ such that $\abs{\varepsilon}\leq 1/2$ we see that the term within the brackets is bounded away from zero. Thus $\abs{k' \det M} \gtrsim \exp(\abs{\Im k'}\sigma)$ and (\ref{eq: Preliminary Resolvent Estimate}) imply
  \begin{equation}\nonumber
    \normLzwei{u_n} \lesssim \normLzwei{f_n} \text{ uniformly in } n.
  \end{equation}

\subsubsection{Regime (iv): $c s \leq k < s$}
As in the previous regime 
  \begin{equation}\nonumber
    k' = k \left(1 - ia_0 sk^{-2} + O(s^{-2})\right) .
  \end{equation}
In particular $k'/k = 1 + O(s^{-1})\rightarrow 1$ and $\Im k' = -a_0 s k^{-1} + O(s^{-1})$ is bounded away from $0, +\infty$ and $-\infty$. Thus 
  \begin{align*}
    \det M &= \frac{1}{k'} \left[ \sin(k'\sigma)\cos(k(1-\sigma) + \cos(k'\sigma)\sin(k(1-\sigma))) \right] + O(s^{-2})   \\
           &= \frac{\sin(k + (k'-k)\sigma)}{k'} + O(s^{-2}) .
  \end{align*}
This implies that $\abs{k'\det M}\approx 1$. Thus from (\ref{eq: Preliminary Resolvent Estimate}) we deduce
  \begin{equation}\nonumber
    \normLzwei{u_n} \lesssim \normLzwei{f_n} \text{ uniformly in } n.
  \end{equation}

\subsection{Conclusion}\label{sec: Conclusion}
Let $u_n$ solve $P_n(s)u_n(x) = f_n(x)$, where $P_n(s)$ is defined in (\ref{eq: separated Stationary Wave Operator}). Section \ref{Regimes where s can live} together with Lemma \ref{thm: forbidden s} shows that the estimate $\normLzwei{u_n} \lesssim s^{1/2}\normLzwei{f_n}$ holds uniformly for any $n$. Therefore, Lemma \ref{thm: P to Pn Lemma} implies Theorem \ref{thm: Maintheorem}.

\section{Exact decay rate for the damped wave equation}\label{sec: Exact Decay Rate for the Damped Wave Equation}
Now we want to prove Theorem \ref{thm: goal}. Therefore recall the definition of the energy $E$ and the damped wave operator $A$ from Section \ref{sec: Introduction}. Then \citep[Theorem 2.4]{BorichevTomilov2010} together with \citep[Proposition 1.3]{BattyDuyckaerts2008} restricted to our situation says in particular that for any $\alpha>0$
  \begin{equation}\label{eq: E vs A}
    \sup_{\norm{U_0}_{D(A)} = 1} E(t,U_0)^{\frac{1}{2}} \approx t^{-\frac{1}{\alpha}} \Leftrightarrow \norm{(is + A)^{-1}} \approx s^{\alpha} .
  \end{equation}
In \citep[Proposition 2.4]{AnantharamanLeautaud2014} it was shown in particular that 
  \begin{equation}\label{eq: A vs P}
    \norm{(is + A)^{-1}} \approx s^{\alpha} \Leftrightarrow \OpNormLzwei{P(s)^{-1}} \approx s^{\alpha - 1}.
  \end{equation}
Actually this equivalence is stated there with `$\approx$` replaced by `$\lesssim$`. But the `$\gtrsim$`-version is included in \citep[Lemma 4.6]{AnantharamanLeautaud2014}. In the appendix of \citep{AnantharamanLeautaud2014} St\'{e}phane Nonnenmacher proved
  \begin{Proposition}[Nonnenmacher, 2014]\label{thm: Nonnenmacher}
    The spectrum of $A$ contains an infinite sequence $(z_j)$ with $\Im z_j \rightarrow \infty$ such that $0 < \Re z_j \lesssim (\Im z_j)^{-3/2}$. 
  \end{Proposition}
Actually he proved this theorem under periodic boundary conditions, but the proof applies also to Dirichlet or Neumann boundary conditions. Note that Proposition \ref{thm: Nonnenmacher} together with (\ref{eq: A vs P}) establishes the `$\gtrsim$`-inequality of Theorem \ref{thm: Maintheorem}.

Using (\ref{eq: E vs A}) and (\ref{eq: A vs P}) together with Theorem \ref{thm: Maintheorem} yields Theorem \ref{thm: goal}.


\bibliographystyle{plainnat}
\bibliography{Refs}

\hfill

Fachrichtung Mathematik, Institut f\"{u}r Analysis, Technische Universit\"{a}t Dresden, 01062, Dresden, Germany. Email: \textit{Reinhard.Stahn@tu-dresden.de}

\end{document}